\documentclass[aps,   amsmath,   amssymb,   amsfonts,   showpace,   twocolumn,   pra]{revtex4}

\usepackage{amscd}
\usepackage{amsthm}
\usepackage{graphicx}
\usepackage{mathdots}
\usepackage[
  colorlinks,
  linkcolor = blue,
  citecolor = blue,
  urlcolor = blue]{hyperref}

\usepackage{amsmath,   amssymb,   amsthm}

\usepackage{tikz}
\usetikzlibrary{arrows}
\usetikzlibrary{calc,   positioning}

\usepackage{thm-restate}

\usepackage{colonequals}

\usepackage[font=small]{caption}

\newcommand{\ket}[1]{|#1\rangle}

\newtheoremstyle{noCaption}% hname
{\topsep}% Space above
{\topsep}% Space below
{\itshape}% Body font
{}% Indent amount
{}% Theorem head font
{}% Punctuation after theorem head
{0pt}% Space after theorem head
{}%

\newtheorem{theorem}{Theorem}

\begin{document}
\title{Positive-partial-transpose distinguishability for lattice-type maximally entangled states}
\author{Zong-Xing Xiong$^{1}$,   Mao-Sheng Li$^{2}$,   Zhu-Jun Zheng$^{1}$,   Chuan-Jie Zhu$^{3}$,  Shao-Ming Fei$^{4,5}$}

\affiliation
{
{\footnotesize  {$^1$Department of Mathematics,
        South China University of Technology,   Guangzhou
        510640,  China}}\\
{\footnotesize  {$^2$Department of Mathematical of Science,
				Tsinghua  University,   Beijing
				100084,  China}}\\
{\footnotesize  {$^3$Department of Physics,
				Renmin University of China,   Beijing
				100872,  China}}\\
{\footnotesize {$^{4}$ School of Mathematical Sciences, Capital Normal University, Beijing
                100048, China}}\\
{\footnotesize  {$^{5}$ Max-Planck-Institute for Mathematics in
                the Sciences, Leipzig 04103, Germany}} \\
}

\begin{abstract}
We study the distinguishability of a particular type of maximally entangled states -- the ``lattice states" using a new approach of semidefinite program. With this, we successfully construct all sets of four ququad-ququad orthogonal maximally entangled states that are locally indistinguishable and find some curious sets of six states having interesting property of distinguishability. Also, some of the problems arose from \cite{CosentinoR14} about the PPT-distinguishability of ``lattice" maximally entangled states can be answered.

\end{abstract}

\pacs{03.  67.  Hk,   03.  65.  Ud }
\maketitle

\section{Introduction}
There is much interest in quantum information to understand the power and limitation of the set of quantum operations and measurements within the paradigm of local operations and classical communication (LOCC).  It is of inherent interest as a tool to understand entanglement and has potential applications in quantum cryptograph,   quantum communication and quantum algorithms.  One of the most basic problem to study LOCC is the challenge of distinguishing certain set of pure states. The problem of state distinguishing serves to explore fundamental question related to local access of global information and the relation between entanglement and locality,   which is of central interest in quantum information theory.  In \cite{Walgate00},   Walgate et.  al.  showed that any two orthogonal states can always be locally distinguished,   no matter whether the states are entangled or not.  In \cite{Bennett99},   Bennett et.  al.   showed that there exist a basis of product states that cannot be distinguished by LOCC,   which is known as ``nolocality without entanglement".  Later,   Horodecki et.  al.  \cite{Horodecki03}  showed a phenomenon of ``more nonlocality with less entanglement".  These facts demonstrate that there is no simple relation between entanglement and locality.

Although not an essential feature of indistinguishable sets of states,   entanglement seems to makes distinguhishability harder in some sense.  Following the result for two state by \cite{Walgate00},   works had been done for search of orthogonal maximally entangled states in higher dimensional systems that cannot be distinguished with LOCC \cite{Ghosh04,Fan04,Ghosh11,Ghosh01,Duan14,Singal2017}.  In \cite{Nathanson05},   Nathanson proved that it is possible to distinguish any $3$ orthogonal maximally entangled states in $\mathbb{C}^3\otimes\mathbb{C}^3$ with LOCC.  What's more,   they proved that any $k > d$ orthogonal maximally entangled states on $\mathbb{C}^d\otimes\mathbb{C}^d$ cannot be distinguished by LOCC.  So it's then an interesting question to ask whether there exist sets of $k \leq d$ orthogonal maximally entangled states that are not perfectly distinguishable by LOCC.  The first example was given by Yu et.  al.   \cite{Duan11},   who proposed a set of four ququad-ququad maximally entangled states in $\mathbb{C}^4\otimes\mathbb{C}^4$ that cannot be perfectly distinguished by positive partial transport (PPT) opeartions. This result is then generalized in \cite{Cosentino13,CosentinoR14} where $k \leq d$ ``lattice" maximally entangled states were constructed in $\mathbb{C}^d\otimes\mathbb{C}^d$ for $d=2^t$.

In this paper,  we study further the problem of distinguishing orthogonal ``lattice" maximally entangled states, which are tensor product of Bell states in a bipartite system. To do this, we first propose a different approach of semidefinete program than that of \cite{Cosentino13}. With this, we construct all sets of four ququad-ququad maximally entangled states that are PPT-indistinguishable and we find some interesting sets of six having interesting property of distinguishability. Also using this new approach, the problem whether ``genuine small" sets of lattice maximally entangled states that are PPT-indistinguishable exist, which arose from \cite{CosentinoR14}, can be answered (for case $t=3$ and $t=4$).

The rest of the paper is organized as follows: Section II presents a summary of the necessary background.  In section III we study the problem of PPT-distinguishability of lattice maximally entangled states using semidefinite program and we propose a new approach which will derive all the upcoming results. We study the problem of distinguishing four orthogonal ququad-ququad maximally entangled states in section IV and in Section V we will discuss the more general case -- the lattice maximally entangled states with $t\geq 3$. We give the conclusion about this paper in Section VI. There is also an Appendix, which gives details about proof of the theorems in Section IV and Section V.

\section{Preliminaries}
Let $\overrightarrow{v}=(v_1,   \cdots,   v_t) \in \{0,   1,   2,   3\}^t$ be a t-dimensional vector and let $|\chi_{\overrightarrow{v}}\rangle \in \mathbb{C}^{2^t}\otimes\mathbb{C}^{2^t}$ be the state given by the tensor product of Bell basis indexed by the vector $\overrightarrow{v}$,    namely,
$\ket{\chi_{\overrightarrow{v}}}=\ket{\psi_{v_1}}_{\mathcal{A}_1\mathcal{B}_1}\otimes \cdots \otimes \ket{\psi_{v_t}}_{\mathcal{A}_t\mathcal{B}_t}$ with $\mathcal{A}_1,  \cdots,  \mathcal{A}_t,  \mathcal{B}_1,  \cdots,  \mathcal{B}_t = \mathbb{C}^2$ and
\begin{align*}
|\psi_0\rangle=\frac{|00\rangle+|11\rangle}{\sqrt{2}}, \ |\psi_1\rangle=\frac{|01\rangle+|10\rangle}{\sqrt{2}},   \\
|\psi_2\rangle=\frac{|00\rangle-|11\rangle}{\sqrt{2}}, \ |\psi_3\rangle=\frac{|01\rangle-|10\rangle}{\sqrt{2}}
\end{align*}
are the standard Bell basis on $\mathbb{C}^{2}\otimes\mathbb{C}^{2}$. Such states $|\chi_{\overrightarrow{v}}\rangle$ are maximally entangled on $\mathcal{A} \otimes \mathcal{B}$,  where $\mathcal{A} = \mathcal{A}_1 \otimes \cdots \otimes \mathcal{A}_t$ and $\mathcal{B} = \mathcal{B}_1 \otimes \cdots \otimes \mathcal{B}_t$. At situation $t=2$ where $\mathcal{A} = \mathcal{A}_1 \otimes \mathcal{A}_2$ and $\mathcal{B} = \mathcal{B}_1 \otimes \mathcal{B}_2$,  the authous of \cite{Duan11} first called these states ``ququad-ququad" and we will follow their usage in this paper. When $t \geq 3$,  we may call them ``qucube-qucube" maximally entangled states,  where cube stands for cube when $t=3$ or hypercube for $t>3$,  or simply ``lattice states" following \cite{Cosentino13}. In what follows we will some times refer to maximally entangled states as ``MESs" for abbreviation.

Let ${\bf Herm}(\mathcal{A}\otimes \mathcal{B})$ and ${\bf Pos}(\mathcal{A}\otimes \mathcal{B})$ denote the sets of all Hermitian operators and positive semi-definite operators on $\mathcal{A}\otimes \mathcal{B}$ respectively.  We say that $M \geq N$ if $M - N$ is positive semi-definite for any Hermitian operators $M$ and $N$.  Let $T_\mathcal{A}(-)$ be the partial transpose map $T\otimes I_\mathcal{B}$ from $\mathcal{A}\otimes \mathcal{B}$ to $\mathcal{A}\otimes \mathcal{B}$,   where $T$ is the transpose map from $\mathcal{A}$ to $\mathcal{A}$,   $I_\mathcal{B}$ is the identity operator on $\mathcal{B}$.   We call a positive semi-definite operator $M\in \mathcal{A}\otimes \mathcal{B}$ a PPT operator if $T_\mathcal{A}(M)\geq 0$.   By ${\bf PPT}(\mathcal{A}: \mathcal{B})$ we denote the set of all PPT operators on the tensor product space $\mathcal{A}\otimes \mathcal{B}$.

\section{Semidefinite program for PPT distinguishability of the lattice states}
Two important properties about the partial transpose map $T_\mathcal{A}(-)$ will be used in the following discussions:  First,    $Tr(T_\mathcal{A}(X))=Tr(X)$ for all operators $X$ on $\mathcal{A}\otimes \mathcal{B}$,   namely,   partial transpose is trace-preseving;  Second,    $(T_\mathcal{A}(H))^\dag=T_\mathcal{A}(H)$ for all $H^\dag=H$,   that is,   partial transpose preserves Hermiticity.

Consider a set of pairwise orthogonal pure states $\{|\chi_1\rangle,   \cdots,   |\chi_k\rangle\}$ where $\langle\chi_i|\chi_j\rangle=\delta_{ij}\ (\forall i, j \in \{1, 2, \cdots, k\})$. Denote $\rho_j=|\chi_j\rangle\langle\chi_j|$ the corresponding density matrix of $|\chi_j\rangle$. Then $\{|\chi_1\rangle,   \cdots,   |\chi_k\rangle\}$ is called PPT-distinguishable if there exist PPT POVMs $P_1,   \cdots,   P_k$ such that $\langle\chi_j|P_i|\chi_j\rangle=\delta_{ij}$,    that is,  $\frac1d\sum_{i = 1}^k \langle P_i,   \rho_i \rangle = 1$,   where $\langle X,   Y \rangle = Tr(X^{\dag}Y)$.
Otherwise,    the set $\{|\chi_1\rangle,   \cdots,   |\chi_k\rangle\}$ is said to be PPT-indistinguishable. The maximum probability of distinguishing a set of states $\{|\chi_1\rangle,   \cdots,   |\chi_k\rangle\}$ by PPT measurements can be expressed as the optimal value of the following semidefinite program \cite{Cosentino13}:

\begin{equation}\begin{split}
\alpha  &=  \max_{P_1,   \cdots,   P_k}\phantom{=}\sum_{j=1}^k \frac1k \langle P_j,   \rho_j \rangle \\
                   &     s.  t.   \phantom{=}P_1 + \cdots + P_k = I_{\mathcal{A}} \otimes I_{\mathcal{B}}\\
                             &\phantom{s. t. }\phantom{=}P_1,   \cdots,   P_k \in {\bf PPT}(\mathcal{A}: \mathcal{B})
\end{split}\end{equation}
We denote the optimal value as $\alpha$ and $\{|\chi_1\rangle\,   \cdots,   |\chi_k\rangle\}$ are PPT-indistinguishable if and only if $\alpha < 1$. The dual problem as what follows can be easily obtained by simple calculation \cite{Cosentino13} and we denote the optimal value as $\beta$:

\begin{equation}\begin{split}
\beta  &=  \min_{Y,   Q_1,   \cdots,   Q_k}\phantom{=} \frac1k  Tr(Y) \\
                   &     s.  t.   \phantom{=}Y - \rho_j \geq T_\mathcal{A} (Q_j),    \   j=1,   \cdots,   k\\
                             &\phantom{s. t. }\phantom{=}Y \in {\bf Herm}(\mathcal{A} \otimes \mathcal{B}) \\
                             &\phantom{s. t. }\phantom{=}Q_1,   \cdots,   Q_k \in {\bf Pos}(\mathcal{A} \otimes \mathcal{B})
\end{split}\end{equation}
By the Slater's condition,   we know that the strong duality holds for this problem and so $\alpha = \beta$.  In \cite{Cosentino13},   the authors further tighten the dual problem by imposing equality instead of inequality constraints in the dual problem:
\begin{equation}\begin{split}
\beta'  &=  \min_{Y,   Q_1,   \cdots,   Q_k}\phantom{=} \frac1k  Tr(Y) \\
                &        s.  t.   \phantom{=}Y - \rho_j = T_\mathcal{A} (Q_j),    \   j=1,   \cdots,   k\\
                             &\phantom{s. t. }\phantom{=}Y \in {\bf Herm}(\mathcal{A} \otimes \mathcal{B}) \\
                             &\phantom{s. t. }\phantom{=}Q_1,   \cdots,   Q_k \in {\bf Pos}(\mathcal{A} \otimes \mathcal{B})
\end{split}\end{equation}
Partially transposing the above program and eliminating the variables $Q_j$'s yields:

\begin{equation}\begin{split}
\beta'  &=  \min_{Y',   Q_1,   \cdots,   Q_k}\phantom{=} \frac1k  Tr(Y') \\
                    &    s.  t.   \phantom{=}Y' \geq T_\mathcal{A} (\rho_j),    \   j=1,   \cdots,   k\\
                             &\phantom{s. t. }\phantom{=}Y' \in {\bf Herm}(\mathcal{A} \otimes \mathcal{B})
\end{split}\end{equation}
where $Y' = T_\mathcal{A}(Y)$. It's clear that $\beta \leq \beta'$. The authors then used this tightened version to construct $d=2^n$ maximally entangled states in $\mathbb{C}^d\otimes\mathbb{C}^d$ which are PPT-indistinguishable,  with the minimum $\beta' <1$. In \cite{Li2015},  examples of $d$ PPT-indistinguishable maximally entangled states in $\mathbb{C}^d\otimes\mathbb{C}^d$ for any $d \geq 4$ are further constructed,  using also the semidefinite program (4). Notice that the dual program (2) can also be expressed as the following:
\begin{equation}\begin{split}
\beta & =  \min_{\substack{Y,   Q_1,   \cdots,   Q_k, \\ R_1,   \cdots,   R_k}}\phantom{=} \frac1k  Tr(Y)  \\
                       & s.  t.  \phantom{=}Y - \rho_j = R_j + T_\mathcal{A} (Q_j),     \  j=1,   \cdots,   k \\
                             & \phantom{s. t. }\phantom{=}Y \in {\bf Herm}(\mathcal{A} \otimes \mathcal{B}) \\
                             & \phantom{s. t. }\phantom{=}Q_1,   \cdots,   Q_k,   R_1,   \cdots,   R_k \in {\bf Pos}(\mathcal{A} \otimes \mathcal{B})
\end{split}\end{equation}
In program (4), it is the $R_j$'s that have been forced to be zero. Here,  we proposed a different way to exploit the power of semidefinite programing. Instead of zeroing $R_j$'s as in \cite{Cosentino13} and \cite{Li2015}, we let the $Q_j$'s to be zero:
\begin{equation}\begin{split}
\beta''  &=  \min_{Y,   R_1,   \cdots,   R_k}\phantom{=} \frac1k  Tr(Y) \\
                      &  s.  t.  \phantom{=}Y - \rho_j = R_j,    \   j=1,   \cdots,   k  \\
                             &\phantom{s. t. }\phantom{=}Y \in {\bf Herm}(\mathcal{A} \otimes \mathcal{B})  \\
                             &\phantom{s. t. }\phantom{=}R_1,   \cdots,   R_k \in {\bf Pos}(\mathcal{A} \otimes \mathcal{B})
\end{split}\end{equation}
Apparently we have $\beta \leq \beta''$. Eliminating $R_j$'s we get a more readable form:
\begin{equation}\begin{split}
\beta''  &=  \min_{Y,   R_1,   \cdots,   R_k}\phantom{=} \frac1k  Tr(Y) \\
                  &      s.  t.   \phantom{=}Y \geq \rho_j,    \  j=1,   \cdots,   k\\
                             &\phantom{s. t. }\phantom{=}Y \in {\bf Herm}(\mathcal{A} \otimes \mathcal{B})
\end{split}\end{equation}
Since $\rho_j=|\chi_j\rangle\langle\chi_j|
$ and states $\{|\chi_1\rangle\,   \cdots,   |\chi_k\rangle\}$ are mutually orthogonal, $Y \geq \rho_j(j=1, \cdots, k)$ obviously imply $Y \geq \sum_{j=1}^{k}\rho_j$. Therefore the optimal solution of  program (7) would be $Y^* = \sum_{j=1}^k \rho_j$ with the optimal value being $\beta'' = 1$. We thus have the following theorem:
\begin{theorem}
Any set of $k$ orthogonal pure states $\big\{|\chi_1\rangle\,   \cdots,   |\chi_k\rangle\big\}$ can be perfectly distinguished by PPT measurement if and only if the feasible solution $\big\{Y = \sum_{i = 0}^k\rho_i; Q_j = 0; R_j = \sum_{i \neq j}\rho_i  \big| j = 1,   \cdots,   k\big\}$ of semidefinite program (5) is an optimal solution.
\end{theorem}

\begin{proof}
The set $\big\{|\chi_1\rangle\,   \cdots,   |\chi_k\rangle\big\}$ can be perfectly distinguished by PPT measurement if and only if $\alpha = 1$. From the above discussion, in order for $\alpha = \beta < \beta'' = 1$,  the feasible solution $\big\{Y = \sum_{i = 0}^k\rho_i; Q_j = 0; R_j = \sum_{i \neq j}\rho_i \big| j = 1,   \cdots,   k\big\}$(which has objective value 1) of problem (5) must not be an optimal solution. Conversely, if it is not an optimal solution, naturally $\alpha = \beta < 1$
\end{proof}

Now let's see how $T_\mathcal{A}(-)$ acts on the lattice states. The action of the partial transpose on the Bell basis can be easily obtained by routine calculation as:
\begin{align*}
T_\mathcal{A}(\psi_0)=\frac12 I-\psi_2, \ T_\mathcal{A}(\psi_1)=\frac12 I-\psi_3,   \\
T_\mathcal{A}(\psi_2)=\frac12 I-\psi_0, \ T_\mathcal{A}(\psi_3)=\frac12 I-\psi_1.
\end{align*}
where $\psi_i=|\psi_i\rangle\langle\psi_i| \ (i \in \{0,1,2,3\})$ are the corresponding density matrices. That is,
\begin{equation*}
T_\mathcal{A}\begin{bmatrix}
\psi_0 & \psi_1 & \psi_2 & \psi_3
\end{bmatrix} = \begin{bmatrix}
\psi_0 & \psi_1 & \psi_2 & \psi_3
\end{bmatrix} \cdot P
\end{equation*}
where
\begin{equation*}
P = \frac12 \begin{bmatrix}
1 & 1 & -1 & 1 \\
1 & 1 & 1 & -1 \\
-1 & 1 & 1 & 1 \\
1 & -1 & 1 & 1 \\
\end{bmatrix}
\end{equation*}
Notice that any lattice MESs can be written as $\ket{\chi_{\overrightarrow{v}}}_{\mathcal{A}\otimes \mathcal{B}}=\ket{\psi_{v_1}}_{\mathcal{A}_1\otimes \mathcal{B}_1}\otimes \cdots \otimes \ket{\psi_{v_t}}_{\mathcal{A}_t\otimes \mathcal{B}_t}$ where $\overrightarrow{v}=(v_1, \cdots, v_t)\in \{0,1,2,3\}^t$,  then it's obvious that
$$
T_\mathcal{A}\big(\chi_{\overrightarrow{v}}\big)=T_{\mathcal{A}_1}\big(\psi_{v_1}\big)\otimes \cdots \otimes T_{\mathcal{A}_t}\big(\psi_{v_t}\big).
$$
So
\begin{equation*}
T_\mathcal{A}\begin{bmatrix}
\chi_{00\cdots 0} & \cdots & \chi_{33\cdots 3}
\end{bmatrix} = \begin{bmatrix}
\chi_{00\cdots 0} & \cdots & \chi_{33\cdots 3}
\end{bmatrix} \cdot P^{\otimes t}
\end{equation*}
Namely,  $P^{\otimes t}$ is the transition matrix of linear transformation $T_\mathcal{A}(-)$ on $\mathbb{C}^{2^t}\otimes\mathbb{C}^{2^t}$ about basis $\big\{ \chi_{i_1 i_2 \cdots i_t} = \psi_{i_1}\otimes\psi_{i_2} \otimes \cdots \otimes\psi_{i_t}\big\}_{i_1 i_2 \cdots i_t \in \{0, 1, 2, 3\}^t}$

The following theorem that was presented in \cite{Cosentino13} shows that in the case where the set to be distinguished contains only lattice states,   the semidefinite program simplifies remarkably.

\begin{theorem}
If $\rho_1,   \cdots,   \rho_k$ are all lattice states,  the probability of successfully distinguishing them by PPT measurements can be expressed as the optimal value of a linear program.
\end{theorem}

For the clarity of subsequent discussion, we will give a simple proof here. \begin{proof} Let $\Delta: L(\mathbb{C}^2\otimes\mathbb{C}^2) \longrightarrow L(\mathbb{C}^2\otimes\mathbb{C}^2)$ be the quantum operation defined as follows:
$$
\Delta(-)=\sum_{i=0}^3 |\psi_i\rangle\langle\psi_i| (-) |\psi_i\rangle\langle\psi_i|,
$$
namely, the dephasing channel under Bell basis.

Let $\Phi=\Delta^{\otimes t}: L(\mathbb{C}^{2^t}\otimes\mathbb{C}^{2^t}) \longrightarrow L(\mathbb{C}^{2^t}\otimes\mathbb{C}^{2^t})$. Then for any
$$
\sigma = \sum_{\substack{i_1\cdots i_t,   j_1\cdots j_t \\ \in \{0,   1,   2,   3\}^t}} \sigma_{i_1\cdots i_t,   j_1\cdots j_t}|\psi_{i_1}\cdots  \psi_{i_t}\rangle\langle \psi_{j_1}\cdots \psi_{j_t}|
$$
we have
$$
\phantom{=}\Phi\left( \sigma \right) = \sum_{\substack{i_1\cdots i_t \\ \in \{0,   1,   2,   3\}^t}} \sigma_{i_1\cdots i_t,   i_1\cdots i_t}|\psi_{i_1}\cdots \psi_{i_t}\rangle\langle \psi_{i_1}\cdots  \psi_{i_t}|.
$$
It's obvious that $\Phi$ is positive and trace-preserving. It acts invariantly on the lattice states: $\Phi(\chi_{\overrightarrow{v}}) = \chi_{\overrightarrow{v}}$.  Moreover, $\Phi$ commutes with the partial transpose $T_\mathcal{A}$, that is, $\Phi(T_\mathcal{A}(-))=T_\mathcal{A}(\Phi(-))$ \cite{Cosentino13}. For any feasible solution $\{Y,   Q_1,   \cdots,   Q_k,   R_1,   \cdots,   R_k\}$ of program (5), by applying $\Phi$ to the entire program, we can always find another solution $\{\Phi(Y),   \Phi(Q_1),   \cdots,   \Phi(Q_k),   \Phi(R_1),   \cdots,   \Phi(R_k)\}$ consisting only of diagonal operators, while the objective function takes the same value.  Since the $\Phi(Q_j)$'s,   $\Phi(R_j)$'s and $\Phi(Y)$ are all diagonal,  the semidefinte program (5) now becomes a linear program,   with the variables being the diagonal elements $r_{i_1\cdots i_t}^{(j)},   q_{i_1\cdots i_t}^{(j)},   y_{i_1\cdots i_t}(i_1\cdots i_t \in \{0,   1,   2,   3\}^t; j = 1, \cdots, k)$ respectively.  Suppose further that the $k$ lattice MESs to be distinguished are $|\chi_{\overrightarrow{v}^{(j)}}\rangle = |\psi_{v_1^{(j)}}\rangle\otimes\cdots\otimes|\psi_{v_t^{(j)}}\rangle \ (j=1,  \cdots,   k)$ and $\rho_j= |\chi_{\overrightarrow{v}^{(j)}}\rangle \langle\chi_{\overrightarrow{v}^{(j)}}| $,  then the linear program can be easily obtained as the following:
\begin{equation}\begin{split}
\beta  =&  \min  \frac1k \sum_{i_1=0}^3 \cdots \sum_{i_t=0}^3 y_{i_1\cdots i_t}  \\
                        s.  t.   &\phantom{=}y_{i_1\cdots i_t} - r_{i_1\cdots i_t}^{(j)} - \sum_{l_1\cdots l_t}P_{i_1\cdots i_t,  l_1\cdots l_t}^{\otimes t} \cdot q_{l_1\cdots l_t}^{(j)}  \\
                       & \phantom{\phantom{=}y_{i_1\cdots i_t} - r_{i_1\cdots i_t}^{(j)} - \sum_{l_1\cdots l_t}P_{i_1\cdots i_t,  l_1\cdots l_t}^{\otimes t}}  = \delta_{{i_1\cdots i_t},{v_1^{(j)}\cdots v_t^{(j)}}}    \\
                             &\phantom{=}r_{i_1\cdots i_t}^{(j)} \geq 0 ,  q_{i_1\cdots i_t}^{(j)} \geq 0   \\
                             &\phantom{=}(i_1\cdots i_t \in \{0,   1,   2,   3\}^t; \ j=1,  \cdots,   k)
\end{split}\end{equation}
where ``$\delta$" is the kronecker delta and $\delta_{\overrightarrow{v},\overrightarrow{w}} = 1$ if and only if $\overrightarrow{v}=\overrightarrow{w}$.
\end{proof}

In what follows we will combine Theorem 1 and Theorem 2 to search sets of $k$ orthogonal lattice MESs in $\mathbb{C}^{2^t}\otimes\mathbb{C}^{2^t}$ that are PPT-indistinguishable. For convenience of analysis we will write linear program (8) in standard matrix form. Denote
\begin{align*}
A = \small \arraycolsep = 0.2em   \left[
\begin{array}{cc|cccc|cccc}
 I & -I & -I &  0 & \cdots & 0  &-P^{\otimes t}&    0   & \cdots &   0  \\
 I & -I & 0  & -I & \cdots & 0  &    0   &-P^{\otimes t}& \cdots &   0  \\
 \vdots & \vdots & \vdots & \vdots & \ddots & \vdots & \vdots & \vdots & \ddots &   \vdots \\
 I & -I & 0  &  0 & \cdots & -I &    0   &    0   & \cdots &-P^{\otimes t}  \\
\end{array}
\right],
\end{align*}
\begin{align*}
\overrightarrow{b} =
\left[
\begin{array}{c}
\overrightarrow{b_1} \\
\overrightarrow{b_2} \\
\vdots \\
\overrightarrow{b_k}
\end{array}
\right]  \bigg(\small \text{where} \ \overrightarrow{b_j} =
\left[\begin{array}{c}
0 \\
\vdots \\
0 \\
1 \\
0 \\
\vdots\\
0 \\
\end{array}\right]
\begin{array}{@{}l}
\\
\phantom{\vdots}\\
\\
\leftarrow  v_1^{(j)}\cdots v_t^{(j)} th\ row \\
\ \ \ \ \ \ \ \ ( j=1, \cdots , k )\\
\phantom{\vdots}\\
{4^t\times 1}\\
\end{array}\bigg) ,
\end{align*}
\begin{equation*}
\overrightarrow{x}=
\small \left[
\begin{array}{c}
\overrightarrow{\alpha}\\
\overrightarrow{\beta}\\
\hline
\overrightarrow{r}^{(1)}\\
\overrightarrow{r}^{(2)}\\
\vdots\\
\overrightarrow{r}^{(k)}\\
\hline
\overrightarrow{q}^{(1)}\\
\overrightarrow{q}^{(2)}\\
\vdots\\
\overrightarrow{q}^{(k)}\\
\end{array}
\right] \normalsize \ \text{and} \ \ \ \overrightarrow{c}= \frac1k \cdot
\small \left[
\begin{array}{c}
\overrightarrow{\mathbf{1}}\\
\overrightarrow{-\mathbf{1}}\\
\hline
\overrightarrow{\mathbf{0}}\\
\overrightarrow{\mathbf{0}}\\
\vdots\\
\overrightarrow{\mathbf{0}}\\
\hline
\overrightarrow{\mathbf{0}}\\
\overrightarrow{\mathbf{0}}\\
\vdots\\
\overrightarrow{\mathbf{0}}\\
\end{array}
\right]
\end{equation*}
\normalsize
where all block matrices or block vectors have dimension $4^t$ and $\overrightarrow{r}^{(j)}=[r_{00\cdots 0}^{(j)}, \cdots, r_{33\cdots 3}^{(j)}]^T,  \overrightarrow{q}^{(j)}=[q_{00\cdots 0}^{(j)}, \cdots, q_{33\cdots 3}^{(j)}]^T(j=1,\cdots,k)$, $\overrightarrow{y} = \overrightarrow{\alpha} - \overrightarrow{\beta} = [\alpha_{00\cdots 0}-\beta_{00\cdots 0}, \cdots, \alpha_{33\cdots 3}-\beta_{33\cdots 3}]^T$ are vectors of the corresponding variables. $\overrightarrow{\mathbf{0}}$, $\overrightarrow{\mathbf{1}}$ represent $4^t \times 1$ vectors whose elements are all $0$'s or $1$'s respectively. Note that by ``$v_1^{(j)}\cdots v_t^{(j)} th$" we mean quaternary number (ranging from $0$ to $4^t-1$). Moreover,  we denote by $\overrightarrow{x} \geq 0$ all elements of $\overrightarrow{x}$ being nonnegative. Therefore, we have
\begin{equation}\begin{split}
\beta  &= \frac1k\ \min_{\overrightarrow{x}} \overrightarrow{c}^T \cdot \overrightarrow{x}\\
          & s.  t.  \phantom{=} A\overrightarrow{x} = \overrightarrow{b} \\
                &\phantom{s. t. } \phantom{=} \overrightarrow{x} \geq 0
\end{split}\end{equation}
which is exactly the same program as (8).

Now by Theorem 1,   $\alpha = \beta = 1$ if and only if the following feasible solution of linear program (9) is an optimal solution:
\begin{equation}\left\{
\begin{array}{l}
\overrightarrow{\alpha}=\sum_{i=1}^k\overrightarrow{b_i}  \\
\overrightarrow{\beta}=\overrightarrow{\mathbf{0}}  \\
\overrightarrow{q}^{(j)}=\overrightarrow{\mathbf{0}}  \\
\overrightarrow{r}^{(j)}=\sum_{i\neq j}\overrightarrow{b_i} \ \ \ \ \ ( j=1, \cdots , k ). \\
\end{array} \right.
\end{equation}

It's obvious that $\text{rank}(A)=k \cdot 4^t$.  Notice that the columns in $A$ corresponding to the following $k\cdot 4^t$ variables (which contain all the nonzero ones in (10)):
\begin{equation*}\begin{split}
& \alpha_{v_1^{(j)} \cdots v_t^{(j)}},  \\
& r_{i_1 \cdots i_t \neq v_1^{(j)} \cdots v_t^{(j)}}^{(j)} \ \ \ \ \ (j=1, \cdots , k)
\end{split}\end{equation*}
form an invertible matrix.
Denote $M$ this matrix and $N$ the matrix consisting of the rest columns (ignoring the order). $M$ and its inverse can be obtain as what follows:
\begin{equation*}
M = \small \arraycolsep = 0.1em  \left[\begin{array}{cccc|cccc}
\overrightarrow{b_1} & \overrightarrow{b_2} & \cdots & \overrightarrow{b_k} &  -I_{:\widehat{\overrightarrow{v}^{(1)}}}  &   0    &  \cdots  &   0    \\
\overrightarrow{b_1} & \overrightarrow{b_2} & \cdots & \overrightarrow{b_k} &   0    &  -I_{:\widehat{\overrightarrow{v}^{(2)}}}  &  \cdots  &   0    \\
\vdots & \vdots & \ddots & \vdots & \vdots & \vdots  & \ddots  & \vdots   \\
\overrightarrow{b_1} & \overrightarrow{b_2} & \cdots & \overrightarrow{b_k} &   0   &   0  &   \cdots  &  -I_{:\widehat{\overrightarrow{v}^{(k)}}}
\end{array}\right]
\end{equation*}
and
\begin{equation*}
M^{-1} = \small \left[\begin{array}{cccc}
 \overrightarrow{b_1}^T &            0           &        \cdots          &        0               \\
          0             & \overrightarrow{b_2}^T &        \cdots          &        0               \\
        \vdots          &         \vdots         &        \ddots          &       \vdots           \\
          0             &            0           &        \cdots          & \overrightarrow{b_k}^T \\
\hline
-I_{\widehat{\overrightarrow{v}^{(1)}}:} &  B^{(1)}_2  &  \cdots  &  B^{(1)}_k   \\
  B^{(2)}_1     & -I_{\widehat{\overrightarrow{v}^{(2)}}:} & \cdots &  B^{(2)}_k   \\
  \vdots     &  \vdots  & \ddots &  \vdots     \\
  B^{(k)}_1     &  B^{(k)}_2  & \cdots  & -I_{\widehat{\overrightarrow{v}^{(k)}}:}
\end{array}\right]
\end{equation*}
where $B_{l}^{(i)}(i\neq l)$ is the $(4^t-1)\times 4^t$ matrix obtained by firstly replacing the $v_1^{(l)}\cdots v_t^{(l)}$th row of $\mathbf{0}_{4^t\times 4^t}$ as $\overrightarrow{b_l}^T$ and secondly deleting the $v_1^{(i)}\cdots v_t^{(i)}$th row. At here and what follows in this paper, by subscript index $_{:\overrightarrow{v}^{(j)}}$ we mean the $v_1^{(j)}\cdots v_t^{(j)}$th (quaternary number) column of a $4^t\times 4^t$ matrix, by  $_{\overrightarrow{v}^{(j)}:}$ we mean the $v_1^{(j)}\cdots v_t^{(j)}$th row, by $_{:\widehat{\overrightarrow{v}^{(j)}}}$ ($_{\widehat{\overrightarrow{v}^{(j)}}:}$) we mean deleting the $v_1^{(j)}\cdots v_t^{(j)}$th column (row) and when discussing vectors instead of matrices, we omit the ``:". Denote $K = \{\overrightarrow{v}^{(1)}th, \cdots, \overrightarrow{v}^{(k)}th\}$ the set of quaternary index corresponding to the $k$ lattice MESs $\{\chi_{\overrightarrow{v}^{(1)}}, \cdots, \chi_{\overrightarrow{v}^{(k)}}\}$ to be distinguished, by subscript $_{:K}$ ($_{K:}$) we mean all the $k$ columns (rows) $\{\overrightarrow{v}^{(1)}th, \cdots, \overrightarrow{v}^{(k)}th\}$ of a $4^t\times 4^t$ matrix and $_{:\widehat{K}}$ ($_{\widehat{K}:}$) mean deleting all these $k$ columns (rows) from it.

Let $\overrightarrow{x}_M$,   $\overrightarrow{x}_N$ be the corresponding array of variables.  We use index $m$'s and $n$'s to denote the variables in $\overrightarrow{x}_M$ and $\overrightarrow{x}_N$ respectively. Rewriting the linear constraint in these notation we have
\small \begin{equation}
A \cdot \overrightarrow{x} = \left[ M,   N \right] \cdot \left[
\begin{array}{c}
\overrightarrow{x}_M\\
\overrightarrow{x}_N
\end{array}\right] = \overrightarrow{b}
\end{equation}
\normalsize
So \small \begin{equation}\begin{split}
\overrightarrow{x}_M &= M^{-1} \cdot \overrightarrow{b} - M^{-1} \cdot N \cdot \overrightarrow{x}_N \\
                     &\triangleq \overrightarrow{b'} + N' \cdot \overrightarrow{x}_N
\end{split}\end{equation}
\normalsize
Namely \begin{equation}
x_m = b'_m + \sum_{n=1}^N N'_{mn} x_n
\end{equation} for all $m = 1, \cdots, M$,   where $b'_m$'s and $N'_{mn}$'s are elements of $\overrightarrow{b'}$ and $N'$. Here,  we obscured the number of columns in $M, N$ with the corresponding block matrix $M, N$. According to these analysis,  rewriting the equation $A\overrightarrow{x}=\overrightarrow{b}$ in program (9) we have
\begin{widetext}
\begin{equation}\begin{split}
M \cdot \arraycolsep = 0em \left[
\begin{array}{c}
\alpha_{\overrightarrow{v}^{(1)}} \\
\alpha_{\overrightarrow{v}^{(2)}} \\
\vdots \\
\alpha_{\overrightarrow{v}^{(k)}} \\
\hline
\overrightarrow{r}_{\widehat{\overrightarrow{v}^{(1)}}}^{(1)} \\
\overrightarrow{r}_{\widehat{\overrightarrow{v}^{(2)}}}^{(2)} \\
\vdots \\
\overrightarrow{r}_{\widehat{\overrightarrow{v}^{(k)}}}^{(k)}
\end{array}\right]
= \arraycolsep = 0em \left[\begin{array}{c}
\overrightarrow{b_1} \\
\overrightarrow{b_2} \\
\vdots \\
\overrightarrow{b_k} \\
\end{array}\right] - \arraycolsep = 0.2em  \left[
\begin{array}{c|cccc}
   I_{:\widehat{K}}   & -\overrightarrow{b_1} &           0           &     \cdots            &     0     \\
   I_{:\widehat{K}}   &          0            & -\overrightarrow{b_2} &     \cdots            &     0     \\
   \vdots             &       \vdots          &         \vdots        &     \ddots            &   \vdots  \\
   I_{:\widehat{K}}   &          0            &           0           &     \cdots            & -\overrightarrow{b_k}
\end{array}\right] \cdot \arraycolsep = 0em \left[
\begin{array}{c}
\\
\overrightarrow{\alpha}_{_{\widehat{K}}} \\
\\
\hline
r_{\overrightarrow{v}^{(1)}}^{(1)} \\
r_{\overrightarrow{v}^{(2)}}^{(2)} \\
\vdots \\
r_{\overrightarrow{v}^{(k)}}^{(k)} \\
\end{array}\right]
+ \arraycolsep = 0.2em \left[\begin{array}{c|cccc}
   I    &   P^{\otimes t}    &   0    &  \cdots  &   0    \\
   I    &   0    &   P^{\otimes t}    &  \cdots  &   0    \\
 \vdots & \vdots & \vdots & \ddots & \vdots    \\
   I    &   0    &   0    & \cdots  &   P^{\otimes t}    \\
\end{array}\right] \cdot \arraycolsep = 0em \left[
\begin{array}{c}
\overrightarrow{\beta} \\
\hline
\overrightarrow{q}^{(1)} \\
\overrightarrow{q}^{(2)} \\
\vdots \\
\overrightarrow{q}^{(k)} \\
\end{array}\right]\\
\\
\end{split}\end{equation}

Acting $M^{-1}$ on it yields
\begin{equation}\begin{split}
\arraycolsep = 0em \left[
\begin{array}{c}
\alpha_{\overrightarrow{v}^{(1)}} \\
\alpha_{\overrightarrow{v}^{(2)}} \\
\vdots \\
\alpha_{\overrightarrow{v}^{(k)}} \\
\hline
\overrightarrow{r}_{\widehat{\overrightarrow{v}^{(1)}}}^{(1)} \\
\overrightarrow{r}_{\widehat{\overrightarrow{v}^{(2)}}}^{(2)} \\
\vdots \\
\overrightarrow{r}_{\widehat{\overrightarrow{v}^{(k)}}}^{(k)}
\end{array}\right] = \arraycolsep = 0em \left[
\begin{array}{c}
1 \\
1 \\
\vdots \\
1 \\
\hline
\sum_{i \neq 1}(\overrightarrow{b_i})_{\widehat{\overrightarrow{v}^{(1)}}} \\
\sum_{i \neq 2}(\overrightarrow{b_i})_{\widehat{\overrightarrow{v}^{(2)}}} \\
\vdots \\
\sum_{i \neq k}(\overrightarrow{b_i})_{\widehat{\overrightarrow{v}^{(k)}}} \\
\end{array}\right] &+ \arraycolsep = 0em \left[\begin{array}{c|cccc}
 0  &  & I_{k\times k} &     &      \\
\hline
I_{\widehat{\overrightarrow{v}^{(1)}}, \widehat{K}} &  0  & (\overrightarrow{b_{2}})_{\widehat{\overrightarrow{v}^{(1)}}} & \cdots & (\overrightarrow{b_{k}})_{\widehat{\overrightarrow{v}^{(1)}}}    \\
I_{\widehat{\overrightarrow{v}^{(2)}}, \widehat{K}} & (\overrightarrow{b_{1}})_{\widehat{\overrightarrow{v}^{(2)}}} &  0  & \cdots & (\overrightarrow{b_{k}})_{\widehat{\overrightarrow{v}^{(2)}}}    \\
\vdots & \vdots & \vdots & \ddots & \vdots \\
I_{\widehat{\overrightarrow{v}^{(k)}}, \widehat{K}} & (\overrightarrow{b_{1}})_{\widehat{\overrightarrow{v}^{(k)}}} & (\overrightarrow{b_{2}})_{\widehat{\overrightarrow{v}^{(k)}}} & \cdots &    0   \\
\end{array}\right] \cdot \arraycolsep = 0em \left[
\begin{array}{c}
\\
\overrightarrow{\alpha}_{_{\widehat{K}}} \\
\\
\hline
r_{\overrightarrow{v}^{(1)}}^{(1)} \\
r_{\overrightarrow{v}^{(2)}}^{(2)} \\
\vdots \\
r_{\overrightarrow{v}^{(k)}}^{(k)} \\
\end{array}\right] \\
&+ \arraycolsep = 0em  \left[\begin{array}{c|cccc}
\overrightarrow{b_1}^T & P_{\overrightarrow{v}^{(1)}:}^{\otimes t} &    0    & \cdots  &     0      \\
\overrightarrow{b_2}^T &    0    & P_{\overrightarrow{v}^{(2)}:}^{\otimes t} & \cdots  &     0      \\
\vdots & \vdots  & \vdots & \ddots & \vdots \\
\overrightarrow{b_k}^T &    0    &    0    & \cdots & P_{\overrightarrow{v}^{(k)}:}^{\otimes t} \\
\hline
-C^{(1)}  & -P_{\widehat{\overrightarrow{v}^{(1)}}:}^{\otimes t} &    D_2^{(1)}    &  \cdots  &     D_k^{(1)}      \\
-C^{(2)}  &    D_1^{(2)}    & -P_{\widehat{\overrightarrow{v}^{(2)}}:}^{\otimes t} &  \cdots  &     D_k^{(2)}      \\
\vdots    &    \vdots    &   \vdots & \ddots &  \vdots \\
-C^{(k)}  &    D_1^{(k)}    &    D_2^{(k)}    &   \cdots  & -P_{\widehat{\overrightarrow{v}^{(k)}}:}^{\otimes t} \\
\end{array}\right] \cdot \arraycolsep = 0em \left[
\begin{array}{c}
\overrightarrow{\beta} \\
\hline
\overrightarrow{q}^{(1)} \\
\overrightarrow{q}^{(2)} \\
\vdots \\
\overrightarrow{q}^{(k)} \\
\end{array}\right]\\
\\
\end{split}\end{equation}
\end{widetext}
where $C^{(i)}$ is the $(4^t-1)\times 4^t$ matrix by firstly replacing rows $\{\overrightarrow{v}^{(1)}th, \cdots , \overrightarrow{v}^{(k)}th\}$ of $I_{4^t\times 4^t}$ as $\overrightarrow{\mathbf{0}}^T$ and secondly deleting the $\overrightarrow{v}^{(i)}$th row. $D_{l}^{(i)}$ is the $(4^t-1)\times 4^t$ matrix by firstly replacing the $\overrightarrow{v}^{(l)}$th row of $\mathbf{0}_{4^t\times 4^t}$ as $P_{\overrightarrow{v}^{(l)}}^{\otimes t}$ and secondly deleting the $\overrightarrow{v}^{(i)}$th row. By letting all the variables in $\overrightarrow{x}_N$ to be zero,   we have solution
\begin{equation}
\overrightarrow{x_0} = \left[\begin{array}{c}
\overrightarrow{b'}\\
\overrightarrow{0}
\end{array}\right]
\end{equation}
corresponding to that of Theorem 1, where
\begin{align*}
(\overrightarrow{x_0})_M = \overrightarrow{b'} = \arraycolsep = 0em \left[
\begin{array}{c}
1 \\
1 \\
\vdots \\
1 \\
\hline
\sum_{i \neq 1}(\overrightarrow{b_i})_{\widehat{\overrightarrow{v}^{(1)}}} \\
\sum_{i \neq 2}(\overrightarrow{b_i})_{\widehat{\overrightarrow{v}^{(2)}}} \\
\vdots \\
\sum_{i \neq k}(\overrightarrow{b_i})_{\widehat{\overrightarrow{v}^{(k)}}}
\end{array}\right]
\end{align*}
and since
\begin{align*}
\overrightarrow{c}_M=
\frac1k\cdot \arraycolsep = 0em \left[\begin{array}{c}
1\\
1\\
\vdots \\
1\\
\hline
\overrightarrow{\mathbf{0}}\\
\overrightarrow{\mathbf{0}}\\
\vdots \\
\overrightarrow{\mathbf{0}}
\end{array}
\right]
\end{align*}
we have the corresponding objective value $\beta'' = z_0 = \sum_{m=1}^M c_m b'_m = 1$.

Also rewriting the objective function we have
\begin{equation}\begin{split}
z &= \overrightarrow{c}_M^T \cdot \overrightarrow{x}_M + \overrightarrow{c}_N^T \cdot \overrightarrow{x}_N \\
  &= \sum_{m=1}^M c_m (b'_m + \sum_{n=1}^N N'_{mn} x_n) + \sum_{n=1}^N c_n x_n \\
  &= \sum_{m=1}^M c_m b'_m + \sum_{n=1}^N (c_n + \sum_{m=1}^M c_m N'_{mn}) x_n \\
  &= z_0 + \sum_{n=1}^N \sigma_n x_n
\end{split}\end{equation}
where
\begin{equation}
\sigma_n = c_n + \sum_{m=1}^M c_m N'_{mn}
\end{equation}
and $N' = \left[ N_{\alpha_{\widehat{K}}},   N_{r_{K}},   N_{\beta},   N_q \right]$ for
\begin{align*}
& \arraycolsep = 0em N_{\alpha_{\widehat{K}}} = \scriptsize \left[\begin{array}{c}
 0   \\
\hline
I_{\widehat{\overrightarrow{v}^{(1)}}, \widehat{K}} \\
I_{\widehat{\overrightarrow{v}^{(2)}}, \widehat{K}} \\
\vdots \\
I_{\widehat{\overrightarrow{v}^{(k)}}, \widehat{K}} \\
\end{array}\right],  \arraycolsep = 0em
N_{r_{K}} = \left[\begin{array}{cccc}
  & I_{k\times k} &     &      \\
\hline
  0  & (\overrightarrow{b_{2}})_{\widehat{\overrightarrow{v}^{(1)}}} & \cdots & (\overrightarrow{b_{k}})_{\widehat{\overrightarrow{v}^{(1)}}}    \\
 (\overrightarrow{b_{1}})_{\widehat{\overrightarrow{v}^{(2)}}} &  0  & \cdots & (\overrightarrow{b_{k}})_{\widehat{\overrightarrow{v}^{(2)}}}    \\
 \vdots & \vdots & \ddots & \vdots \\
 (\overrightarrow{b_{1}})_{\widehat{\overrightarrow{v}^{(k)}}} & (\overrightarrow{b_{2}})_{\widehat{\overrightarrow{v}^{(k)}}} & \cdots &    0   \\
\end{array}\right],  \\
\\
& \arraycolsep = 0em \scriptsize N_{\beta} = \left[\begin{array}{c}
\overrightarrow{b_1}^T \\
\overrightarrow{b_2}^T \\
\vdots \\
\overrightarrow{b_k}^T \\
\hline
-C^{(1)}   \\
-C^{(2)}   \\
\vdots  \\
-C^{(k)}   \\
\end{array}\right],
\arraycolsep = 0em
N_q = \left[\begin{array}{cccc}
 P_{\overrightarrow{v}^{(1)}:}^{\otimes t} &    0    & \cdots  &     0      \\
    0    & P_{\overrightarrow{v}^{(2)}:}^{\otimes t} & \cdots  &     0      \\
 \vdots  & \vdots & \ddots & \vdots \\
    0    &    0    & \cdots & P_{\overrightarrow{v}^{(k)}:}^{\otimes t} \\
\hline
 -P_{\widehat{\overrightarrow{v}^{(1)}}:}^{\otimes t} &    D_2^{(1)}    &  \cdots  &     D_k^{(1)}      \\
    D_1^{(2)}    & -P_{\widehat{\overrightarrow{v}^{(2)}}:}^{\otimes t} &  \cdots  &     D_k^{(2)}      \\
    \vdots    &   \vdots & \ddots &  \vdots \\
    D_1^{(k)}    &    D_2^{(k)}    &   \cdots  & -P_{\widehat{\overrightarrow{v}^{(k)}}:}^{\otimes t} \\
\end{array}\right]
\end{align*}
and

\begin{align*}
& \overrightarrow{c}_{N_{\alpha_{\widehat{K}}}} = \overrightarrow{\mathbf{1}}_{(4^t-k)\times 1},   & \overrightarrow{c}_{N_{r_{K}}} = \overrightarrow{\mathbf{0}}_{k\times 1} \\
\\
& \overrightarrow{c}_{N_\beta} = -\overrightarrow{\mathbf{1}}_{4^t\times 1},   & \overrightarrow{c}_{N_q} = \overrightarrow{\mathbf{0}}_{(k\cdot 4^t)\times 1}
\end{align*}

In order to find another feasible solution $\overrightarrow{x}$ other than $\overrightarrow{x_0}$ having smaller objective value $z < z_0 = 1$,   we need to add some variables $x_n > 0$ other than the original $x_m$'s in $\overrightarrow{x_0}$.  These $n$'s must satisfy:
\begin{equation}
x_m = b'_m + \sum_{\{n|x_n>0\}} N'_{mn} x_n \geq 0 \ \ (m=1,  \cdots ,  M)
\end{equation}
and
\begin{equation}
\sum_{\{n|x_n>0\}} \sigma_n x_n < 0
\end{equation}
That is,  PPT-indistinguishability can be immediately deduced if such $x_n$'s exist, due to Theorem 1.

Whether these $x_n$'s are available depends on $k$ and the set $\big\{\chi_{\overrightarrow{v}^{(1)}}, \cdots , \chi_{\overrightarrow{v}^{(k)}}\big\}$ being distinguished,  for these states are represented in the above linear program by vectors $\overrightarrow{b_1}, \cdots , \overrightarrow{b_k}$,  which affect dominantly as shown by equation (15).

In the next section, we will specialize the above discussion to the case $(t=2, k=4)$ and construct all sets of four PPT-indistinguishable lattice MESs. In section V, we will discuss the PPT-distinguishability of lattice MESs for $t=3$ and $t=4$, where we find out that no ``genuine" small sets of PPT-indistinguishable lattice MESs exist in such cases.

\section{Distinguishability for any four orhogonal ququad-ququad maximally entangled states}

In \cite{Nathanson2013}, the authors find out that any set of three orthogonal ququad-ququad MESs in $\mathbb{C}^4\otimes\mathbb{C}^4$ can be distinguished using only one-way LOCC. The first set of four orthogonal ququad-ququad MESs $\big\{\ket{\psi_0}\otimes\ket{\psi_0},  \ket{\psi_1}\otimes\ket{\psi_1},  \ket{\psi_2}\otimes\ket{\psi_1},  \ket{\psi_3}\otimes\ket{\psi_1}\big\}$ which is LOCC indistinguishable was given in \cite{Duan11}. This result is then generalized to $\mathbb{C}^d\otimes\mathbb{C}^d$ where $d=2^t$ in \cite{Cosentino13}. All of the examples constructed are actually PPT-indistinguishable,  from where  LOCC indistinguishability can be infered. In this section,  we will study the PPT-distinguishability of any set of four orthogonal ququad-ququad MESs,  using the approach of semidefinite program we have discussed in section III. That is, we will consider the specific situation when $t=2$ and $k=4$.

Consider again linear program (9), in this case
\begin{align*}
A = \small \arraycolsep = 0.2em  \left[
\begin{array}{cc|cccc|cccc}
 I & -I & -I &  0 &  0 &  0  &-P^{\otimes 2}&    0   & 0 &   0  \\
 I & -I & 0  & -I &  0 &  0  &    0   &-P^{\otimes 2}& 0 &   0  \\
 I & -I & 0  &  0 & -I &  0  &    0   &   0  &-P^{\otimes 2}& 0 \\
 I & -I & 0  &  0 &  0 & -I  &    0   &    0   & 0 &-P^{\otimes 2}  \\
\end{array}
\right]
\end{align*}
where all block matrices are of dimension $4^2=16$ and so on. Restricting ourself to $k=4$, we can reproduce the analysis in section III until we run into the problem: for which choice of $\{\overrightarrow{b_1}, \overrightarrow{b_2}, \overrightarrow{b_3}, \overrightarrow{b_4}\}$, the replacing variables $x_n$'s are available? We present the following theorem as an answer to this problem:

\begin{theorem}
There are sets of 6 ququad-ququad orthogonal MESs where any 4 states among them are PPT-indistinguishable. These sets are:
\begin{align*}
S_0 = \{\chi_{02}, \chi_{12}, \chi_{20}, \chi_{21}, \chi_{23}, \chi_{32}\},\\
S_1 = \{\chi_{03}, \chi_{13}, \chi_{20}, \chi_{21}, \chi_{22}, \chi_{33}\},\\
S_2 = \{\chi_{00}, \chi_{10}, \chi_{21}, \chi_{22}, \chi_{23}, \chi_{30}\},\\
S_3 = \{\chi_{01}, \chi_{11}, \chi_{20}, \chi_{22}, \chi_{23}, \chi_{31}\},\\
S_4 = \{\chi_{02}, \chi_{12}, \chi_{22}, \chi_{30}, \chi_{31}, \chi_{33}\},\\
S_5 = \{\chi_{03}, \chi_{13}, \chi_{23}, \chi_{30}, \chi_{31}, \chi_{32}\},\\
S_6 = \{\chi_{00}, \chi_{10}, \chi_{20}, \chi_{31}, \chi_{32}, \chi_{33}\},\\
S_7 = \{\chi_{01}, \chi_{11}, \chi_{21}, \chi_{30}, \chi_{32}, \chi_{33}\},\\
S_8 = \{\chi_{00}, \chi_{01}, \chi_{03}, \chi_{12}, \chi_{22}, \chi_{32}\},\\
S_9 = \{\chi_{00}, \chi_{01}, \chi_{02}, \chi_{13}, \chi_{23}, \chi_{33}\},\\
S_{10} = \{\chi_{01}, \chi_{02}, \chi_{03}, \chi_{10}, \chi_{20}, \chi_{30}\},\\
S_{11} = \{\chi_{00}, \chi_{02}, \chi_{03}, \chi_{11}, \chi_{21}, \chi_{31}\},\\
S_{12} = \{\chi_{02}, \chi_{10}, \chi_{11}, \chi_{13}, \chi_{22}, \chi_{32}\},\\
S_{13} = \{\chi_{03}, \chi_{10}, \chi_{11}, \chi_{12}, \chi_{23}, \chi_{33}\},\\
S_{14} = \{\chi_{00}, \chi_{11}, \chi_{12}, \chi_{13}, \chi_{20}, \chi_{30}\},\\
S_{15} = \{\chi_{01}, \chi_{10}, \chi_{12}, \chi_{13}, \chi_{21}, \chi_{31}\}.\\
\end{align*}
Moreover,  any 4 ququad-ququad orthogonal MESs that are PPT-indistinguishable must be contained by one of these sets.
\end{theorem}

For consideration of readability we are omitting the proof here and to leave it to the Appendix. This result is interesting, for any 3 ququad-ququad orthogonal maximally entangled states are LOCC-distinguishable \cite{Nathanson2013}. Thus, we have found sets of 6 ququad-ququad orthogonal maximally entangled states where any 4 states among them are LOCC-indistinguishable (infered by PPT-indistinguishability) while any 3 among them are LOCC-distinguishable.

In \cite{Tian15, Tian2015}, the authors study the local distinguishability of the orthogonal ququad-ququad MESs, in terms of the relationship between MUBs and LOCC. They proved that every set of four such states is either distinguishable by one-way LOCC or LOCC-indistinguishable and any LOCC-indistinguishable set of four is equivalent to $\big\{\ket{\psi_0}\otimes\ket{\psi_0},  \ket{\psi_1}\otimes\ket{\psi_1},  \ket{\psi_2}\otimes\ket{\psi_1},  \ket{\psi_3}\otimes\ket{\psi_1}\big\}$, up to any local unitary transformation and permutation. Restricting ourselves to PPT-distinguishability and using the method of semidefinite program, we independently derived a compatible result here. Since PPT-indistinguishability immediately infers LOCC-indistinguishability, the sets of 6 ququad-ququad MESs we have just found turn out to be some sort of orbits by $\big\{\ket{\psi_0}\otimes\ket{\psi_0},  \ket{\psi_1}\otimes\ket{\psi_1},  \ket{\psi_2}\otimes\ket{\psi_1},  \ket{\psi_3}\otimes\ket{\psi_1}\big\}$ under local unitary transformations and permutations.

\section{PPT distinguishability for orthogonal lattice maximally entangled states in $\mathbb{C}^{2^t}\otimes\mathbb{C}^{2^t}$ with $t=3$ or $t=4$}

It has been proved in \cite{Nathanson2013} that any $k\leq \frac{d}{2}+1$ orthogonal MESs in $\mathbb{C}^d\otimes\mathbb{C}^d$ can be distinguished with PPT measurements. In \cite{CosentinoR14},  the authors constructed 8 orthogonal lattice MESs in $\mathbb{C}^8\otimes\mathbb{C}^8$ and 15 such states in $\mathbb{C}^{16}\otimes\mathbb{C}^{16}$ which are PPT-indistinguishable. More generally, when $d=2^t$,  they successfully constructed $7d/8 < k \leq d$ lattice MESs in $\mathbb{C}^{d}\otimes\mathbb{C}^{d}$ which are PPT-indistinguishable. But it remains puzzled whether sets of $\frac{d}{2}+1 < k \leq 7d/8$ such orthogonal lattice MESs that are PPT-indistinguishable can be constructed. Fortunately, our approach as shown in section III seems powerful to deal with this problem. As what we'll show in the following, the answer tend to be negative.

When $t=3$, we need only to check the case $k=6$ and $k=7$. We conclude here that no such PPT-indistinguishable sets exists:
\begin{theorem}
Every set of 6 or 7 orthogonal lattice MESs in $\mathbb{C}^8\otimes\mathbb{C}^8$ can be distinguished with PPT measurements.
\end{theorem}

Also, we are to present only our results here and to leave the proofs to the Appendix. Notice here that when $t=3$ we can also derive a similar result as Theorem 3 for $k=8$, which can be easily deduced from the proof in the Appendix.

For the case $t=4$, we also have a negative answer:
\begin{theorem}
Every set of $k\leq 14$ orthogonal lattice MESs in $\mathbb{C}^{16}\otimes\mathbb{C}^{16}$ can be distinguished with PPT measurements.
\end{theorem}

\section{Conclusion and Discussion}
In this paper, we study the PPT-distinguishability of lattice maximally entangled states using a different approach of semidefinite program (or linear program since the PPT-distinguishability problem of such sort of ``lattice" states have a linear structure \cite{Cosentino13}) which enable us to figure out all the PPT-indistinguishable sets.

The problem has been left over by \cite{CosentinoR14} whether ``genuine small" sets (with $k<=7d/8$ where $d=2^t$) of lattice maximally entangled states that are PPT-indistinguishable can be found. For cases $t=3$ and $t=4$, we show that such sets are all PPT-dintinguishable so the answer is no. For more general cases where $t > 4$, the problem becomes too tedious for us to discuss here. Nevertheless, we believe that the answer tends to be negative likewise for we've run numerical program to search (stochastically) for examples of these PPT-indistinguishable sets (also using the approach we described in section III but by computer evaluation), but no example emerges after long exhausted computer running. Since these PPT-indistinguishable sets always appear in groups if there were any (just as what we've shown in Theorem 3), the probability for stochastic search to hit them would not be too small and hence they ought to emerge sometimes during our long-run experiment (due to the law of large numbers), if they existed \cite{1}.

\bigskip
\noindent{\bf Acknowledgments}\, \, 
This work was supported by the NSFC through Grant No.
11571119.

\section{Appendix}

\centering {\bf (i) Proof of Theorem 3 }
\begin{proof}
We will discuss in what situation variables $x_n > 0$ other than the original chosen $x_m$'s are available such that equation (19) and (20) satisfies.

For variables $\alpha_{\overrightarrow{v}^{(1)}}, \alpha_{\overrightarrow{v}^{(2)}}, \alpha_{\overrightarrow{v}^{(3)}}, \alpha_{\overrightarrow{v}^{(4)}}$ and  $r_{\overrightarrow{v}^{(l)}}^{(j)}\ (j,l \in \{1, 2, 3, 4\}; l\neq j)$,  they equal to $1$ in $\overrightarrow{x_0}$,   so equations (19) are easy to satisfy,   once we multiply all the $x_n$'s with a very small positive number which will certainly not break equation (20).  So we may only consider the variables in $\overrightarrow{r}_{\widehat{K}}^{(j)}(j\in \{1, 2, 3, 4\})$,  which equal to $0$ in $\overrightarrow{x_0}$. We now list out $\sigma_n$'s for all columns in $N$ and also the $N'_{mn}$ for all rows corresponding to these variables (in matrix form):
\begin{equation}\begin{split}
& \arraycolsep = 0.4em  \begin{array}{|c|c|c|c|c|}
\hline
    &  N_{\alpha_{\widehat{K}}}  &   N_{r_{K}}  & N_{\beta_{K}} & N_{\beta_{\widehat{K}}} \\
\hline
\sigma_n \text{'s} &  \overrightarrow{\mathbf{1}}^T  &  \overrightarrow{\mathbf{1}}^T  & \overrightarrow{\mathbf{0}}^T & -\overrightarrow{\mathbf{1}}^T  \\
\hline
\overrightarrow{r}_{\widehat{_{K}}}^{(1)} & I & 0 & 0 & -I \\
\hline
\overrightarrow{r}_{\widehat{_{K}}}^{(2)} & I & 0 & 0 & -I \\
\hline
\overrightarrow{r}_{\widehat{_{K}}}^{(3)} & I & 0 & 0 & -I \\
\hline
\overrightarrow{r}_{\widehat{_{K}}}^{(4)} & I & 0 & 0 & -I \\
\hline
\end{array}\\
& \arraycolsep = 0.2em \begin{array}{|c|c|c|c|c|}
\hline
    &   N_{q^{(1)}} &  N_{q^{(2)}} & N_{q^{(3)}}  &  N_{q^{(4)}} \\
\hline
\sigma_n \text{'s} &  P_{\overrightarrow{v}^{(1)}:}^{\otimes 2}  &  P_{\overrightarrow{v}^{(2)}:}^{\otimes 2} & P_{\overrightarrow{v}^{(3)}:}^{\otimes 2}  & P_{\overrightarrow{v}^{(4)}:}^{\otimes 2} \\
\hline
\overrightarrow{r}_{\widehat{K}}^{(1)} & -P_{\widehat{K}:}^{\otimes 2} & 0 & 0 & 0 \\
\hline
\overrightarrow{r}_{\widehat{K}}^{(2)} & 0 & -P_{\widehat{K}:}^{\otimes 2} & 0 &  0  \\
\hline
\overrightarrow{r}_{\widehat{K}}^{(3)} & 0 & 0 & -P_{\widehat{K}:}^{\otimes 2} & 0 \\
\hline
\overrightarrow{r}_{\widehat{K}}^{(4)} & 0 & 0 & 0 & -P_{\widehat{K}:}^{\otimes 2}   \\
\hline
\end{array}
\end{split}\end{equation}

It is obvious that the crucial factor about this problem is the structure of matrix $P^{\otimes 2}$. Since
\begin{equation*}
P = \frac12 \begin{bmatrix}
1 & 1 & -1 & 1 \\
1 & 1 & 1 & -1 \\
-1 & 1 & 1 & 1 \\
1 & -1 & 1 & 1 \\
\end{bmatrix},
\end{equation*}
the elements of $P^{\otimes 2}$ consist of only $\frac{1}{4}$ and $-\frac{1}{4}$,  with more positive elements than negative elements at each row (or column since $P^{\otimes 2}$ is symmetric). Specifically,  each row of $P^{\otimes 2}$ has $16$ elements,  $10$ of which are positive and the other $6$ are negative. Moreover,  for any two different rows,  $4$ places among those of the $10$ positive elements at one row take negative on the other row; Conversely,  $4$ places among those of the $6$ negative elements at one row take positive on the other row. This guarantee the orthogonality of matrix $P^{\otimes 2}$.
\begin{align*}
\begin{array}{cc@{}c|c@{}c}
                         &          \multicolumn{2}{c|}{_{10}}       &      \multicolumn{2}{c}{_{6}}       \\
\substack{some\\row:}    & \multicolumn{2}{c|}{\overbrace{++++++++++}} & \multicolumn{2}{c}{\overbrace{------}} \\
\substack{another\\row:} & \underbrace{++++++} & \underbrace{----} & \underbrace{++++} & \underbrace{--} \\
                         &    ^{6}    &      ^{4}          &      ^{4}          &    ^{2}     \\
\end{array} \\
\text{\scriptsize {Any two rows of $P^{\otimes 2}$ (ignore the order of elements)}}
\end{align*}

To add positive variables $x_n \in \overrightarrow{x_N}$ such that $\sum_{\{n|x_n>0\}} \sigma_n x_n < 0$ and $r_m^{(j)} =  \sum_{\{n|x_n>0\}} N'_{mn} x_n \geq 0 (j \in \{1, 2, 3, 4\}; m \notin K)$,  we need only to consider columns in $N_{\alpha_{\widehat{K}}}$ and $N_q=[N_{q^{(1)}}, N_{q^{(2)}}, N_{q^{(3)}}, N_{q^{(4)}}]$. To explain why it is so,  first consider the condition $\sum_{\{n|x_n>0\}} \sigma_n x_n < 0$. At first sight,  we may probably choose from columns $N_{\beta_{\widehat{_{K}}}}$ or from $N_{q^{(j)}}$(where $P_{\overrightarrow{v}^{(j)},  n}^{\otimes 2}<0)$. But further considering $r_m^{(j)} =  \sum_{\{n|x_n>0\}} N'_{mn} x_n \geq 0 (j \in \{1, 2, 3, 4\}; m \notin K )$,  choosing $n \in N_{\beta_{\widehat{K}}}$ will be significantly adverse. Since each column of $P^{\otimes 2}$ has $4$ more elements of positive than negative,  after deleting $4$ rows $K = \{\overrightarrow{v}^{(1)}th, \overrightarrow{v}^{(2)}th, \overrightarrow{v}^{(3)}th, \overrightarrow{v}^{(4)}th\}$,  each column of $-P_{\widehat{K}}^{\otimes 2}$ has negative elements at least as many as the positive ones. Namely,  even if we had fortunately chosen all the columns $n \in N_{q^{(j)}}$ such that for a fixed $n$ elements $P_{\overrightarrow{v}^{(j)},n}^{\otimes 2}$'s are all positive $(j \in \{1, 2, 3, 4\})$,  at least one of the rows $m \notin K$ must satisfy $\sum_{\{n\in N_{q^{(j)}}|x_n>0\}} -P_{mn}^{\otimes 2} x_n \leq 0 $ (actually,  for any $4$ fixed rows $K = \{\overrightarrow{v}^{(1)}th, \overrightarrow{v}^{(2)}th, \overrightarrow{v}^{(3)}th, \overrightarrow{v}^{(4)}th\}$ rows in $P^{\otimes 2}$,  at most one column of $P_{K:}^{\otimes 2}$ can take positive $\frac14$ simultaneously). At this situation,  strict inequality $\sum_{\{n\in N_{q^{(j)}}|x_n>0\}} -P_{mn}^{\otimes 2} x_n < 0$ is unavoidable for some $m \notin K$. Instead,  we should choose columns $n_{\alpha} \in N_{\alpha_{\widehat{K}}}$ in order to satisfy $r_m^{(j)} =  \sum_{\{n|x_n>0\}} N'_{mn} x_n \geq 0 \ (j\in\{1, 2, 3, 4\}; m \notin K )$. Actually,  same value of the variable $n_{\alpha_{\widehat{K}}}=n_{\beta_{\widehat{K}}}$ in same place of $N_{\alpha_{\widehat{K}}}$ and $N_{\beta_{\widehat{K}}}$ will cancel each other in equation (19) and (20) (the $\sigma_n$'s and $N_{mn}$'s are opposite as shown in (21)),  therefore we handle this problem by allowing $x_{n_{\alpha_{\widehat{K}}}}$ to be negative and ignore $N_{\beta_{\widehat{K}}}$. Since $\sigma_n = 1$ for $n \in N_{\alpha_{\widehat{K}}}$,  we may choose also columns $n \in N_{q^{(j)}}(P_{\overrightarrow{v}^{(j)},  n}^{\otimes 2}<0)$ to satisfy $\sum_{\{n|x_n>0\}} \sigma_n x_n < 0$. Now suppose that we take $x_{n_{q^{(j)}}}>0$ for all those $n_{q^{(j)}}$'s such that $P_{\overrightarrow{v}^{(j)}, n_{q^{(j)}}}^{\otimes 2}<0 \ (j\in \{1, 2, 3, 4\})$. Since each row of $P^{\otimes 2}$ has $6$ negative elements,  denoting these variables $x_{q_{1}^{(j)}}, \cdots, x_{q_{6}^{(j)}}$. Without loss of generality we may assume that $x_{q_{1}^{(j)}}= \cdots =x_{q_{6}^{(j)}}=q^{(j)}>0 \ (j \in \{1, 2, 3, 4\})$, so we have
\begin{equation}
\sum_{m \notin K} x_{\alpha_{m}} < \frac14\sum_{j \in \{1, 2, 3, 4\}}6 \cdot q^{(j)}
\end{equation}
Fixing $j \in \{1, 2, 3, 4\}$, for $P_{\widehat{\overrightarrow{v}^{(j)}}}^{\otimes 2}$ on these $6$ columns, each row takes $(4+2-)$,  so we have
\begin{equation}
x_{\alpha_m}-\frac14\cdot 4q^{(j)}+\frac14\cdot 2q^{(j)} \geq 0 \ (\text{for all } m \notin K)
\end{equation}
and thus
\begin{equation}
\sum_{m \notin K} x_{\alpha_{m}} \geq (16-4) \cdot \frac14 \cdot (4-2)q^{(j)} \ (j \in \{1, 2, 3, 4\})
\end{equation}
It's obvious that (22) contradicts (24). So we should take less $n_{q^{(j)}}$'s for those $P_{\overrightarrow{v}^{(j)}, n_{q^{(j)}}}^{\otimes 2}<0$ or take some $n_{q^{(j)}}$ such that $P_{\overrightarrow{v}^{(j)}, n_{q^{(j)}}}^{\otimes 2}>0$. As we have already discuss above,  the latter choice will be adverse, for that will decrease the right hand side of (22) while increase the right hand side of (24). Now suppose that $x_{q_{1}^{(j)}}= \cdots =x_{q_{5}^{(j)}}=q^{(j)}>0 \ (j \in \{1, 2, 3, 4\})$,  so
\begin{equation}
\sum_{m \notin K} x_{\alpha_{m}} < \frac14\sum_{j \in \{1, 2, 3, 4\}}5 \cdot q^{(j)}
\end{equation}
For any $5$ columns chosen like this for a fixed $\overrightarrow{v}^{(j)}$,  10 rows of $P_{\widehat{\overrightarrow{v}^{(j)}}}^{\otimes 2}$ take $(3+2-)$,  5 rows of $P_{\widehat{\overrightarrow{v}^{(j)}}}^{\otimes 2}$ take $(4+1-)$,  because when we deselect $x_{q_{6}^{(j)}}$,  we decrease $1+$ from $(4+2-)$ for 10 rows and $1-$ from $(4+2-)$ for 5 rows in $P_{\widehat{\overrightarrow{v}^{(j)}}}^{\otimes 2}$. So for $m \notin K$ we have
\begin{equation}
 x_{\alpha_m} \geq \frac14\cdot 3q^{(j)}-\frac14\cdot 2q^{(j)}
\end{equation}
or
\begin{equation}
 x_{\alpha_m} \geq \frac14\cdot 4q^{(j)}-\frac14\cdot 1q^{(j)}
\end{equation}
To avoid the contradiction like what between (22) and (24), we should hope more of $m \notin K$ to satisfies (26) rather than (27), therefore columns corresponding to $x_{q_{6}^{(1)}}=x_{q_{6}^{(2)}}=x_{q_{6}^{(3)}}=x_{q_{6}^{(4)}}=0$ in $P_{\overrightarrow{v}^{(1)}}^{\otimes 2}, P_{\overrightarrow{v}^{(2)}}^{\otimes 2}, P_{\overrightarrow{v}^{(3)}}^{\otimes 2}, P_{\overrightarrow{v}^{(4)}}^{\otimes 2}$ respectively should be equal. Moreover, by the 4 rows $K=\{\overrightarrow{v}^{(1)}th, \overrightarrow{v}^{(2)}th, \overrightarrow{v}^{(3)}th, \overrightarrow{v}^{(4)}th\}$ chosen like this, inequality (26) and (27) are independent of the particular $j \in \{1,2,3,4\}$ (for fixed $m \notin K$). This is crucial for we must let $\sum_{m \notin K} x_{\alpha_{m}}$ to be as small as not to contradict (25). In this case,
\begin{equation}\begin{split}
\sum_{m \notin K} x_{\alpha_{m}} \geq 10 \cdot \frac14 \cdot q^{(j)} + (6-4) & \cdot \frac14 \cdot 3q^{(j)} \\
                                                                           & (j\in \{1, 2, 3, 4\})
\end{split}\end{equation}
Now (25) and (28) do not contradict any more and so (19) and (20) can be all satisfied. We may further take less $n_{q^{(j)}}$'s for those $P_{\overrightarrow{v}^{(j)}, n_{q^{(j)}}}^{\otimes 2}<0$, for example, assume that $x_{q_{5}^{(j)}}=x_{q_{6}^{(j)}}=0 \ (j \in \{1, 2, 3, 4\})$. However, this is impossible, because for any two columns of $P^{\otimes 2}$, only 2 rows are like ``$--$", while 6 rows are like ``$++$" and 8 rows are like ``$+-$" hence inequality like (26) and (27) cannot be satisfied by all $j \in \{1, 2, 3, 4\}$ simultaneously and thus contradiction like between (22) and (24) will occur. So what we have found as between (25) and (28) is the only situation that (19) and (20) are satisfied.

By now,  we have found a way to construct four ququad-ququad orthogonal maximally entangled states: choosing any one column from $P^{\otimes 2}$, then 6 of the elements are negative and any 4 rows $\{\overrightarrow{v}^{(1)}th, \overrightarrow{v}^{(2)}th, \overrightarrow{v}^{(3)}th, \overrightarrow{v}^{(4)}th\}$ among these 6 correspond to 4 PPT-indisinguishable states $\{\ket{\chi_{\overrightarrow{v}^{(1)}}}, \ket{\chi_{\overrightarrow{v}^{(2)}}}, \ket{\chi_{\overrightarrow{v}^{(3)}}}, \ket{\chi_{\overrightarrow{v}^{(4)}}}\}$. For example,  the four states presented in \cite{Duan11} $\big\{\ket{\psi_0}\otimes\ket{\psi_0},  \ket{\psi_1}\otimes\ket{\psi_1},  \ket{\psi_2}\otimes\ket{\psi_1},  \ket{\psi_3}\otimes\ket{\psi_1}\big\}$ can be constructed by first choose column $23th$ and then choose row $\{00th, 11th, 21th, 31 th\}$ where the crossing elements all take negative:

\begin{align*}
\arraycolsep = 0.2em \begin{array}{l|c|c|c|c|c|cl}
\cline{5-5}
\text{\scriptsize 00}&++-+&++-+&--+&-&&++-+& \leftarrow\ket{\chi_{00}}\\
\text{\scriptsize 01}&+++-&+++-&---&+&&+++-\\
\text{\scriptsize 02}&-+++&-+++&+--&-&&-+++& \leftarrow\text{\small negative}\\
\text{\scriptsize 03}&+-++&+-++&-+-&-&&+-++& \leftarrow\text{\small negative}\\
\cline{1-7}
\text{\scriptsize 10}&++-+&++-+&++-&+&&--+-\\
\text{\scriptsize 11}&+++-&+++-&+++&-&&---+& \leftarrow\ket{\chi_{11}}\\
\text{\scriptsize 12}&-+++&-+++&-++&+&&+---\\
\text{\scriptsize 13}&+-++&+-++&+-+&+&&-+--\\
\cline{1-7}
\text{\scriptsize 20}&--+-&++-+&++-&+&&++-+\\
\text{\scriptsize 21}&---+&+++-&+++&-&&+++-& \leftarrow\ket{\chi_{21}}\\
\text{\scriptsize 22}&+---&-+++&-++&+&&-+++\\
\text{\scriptsize 23}&-+--&+-++&+-+&+&&+-++\\
\cline{1-7}
\text{\scriptsize 30}&++-+&--+-&++-&+&&++-+\\
\text{\scriptsize 31}&+++-&---+&+++&-&&+++-& \leftarrow\ket{\chi_{31}}\\
\text{\scriptsize 32}&-+++&+---&-++&+&&-+++\\
\text{\scriptsize 33}&+-++&-+--&+-+&+&&+-++\\
\cline{5-5}
\multicolumn{4}{c}{}&\multicolumn{1}{c}{\uparrow}  \\
\multicolumn{4}{c}{}&\multicolumn{1}{c}{\text{\small 23}} \\
\end{array}
\end{align*}

Actually,  since conditions (19) and (20) are necessary and sufficient (by Theorem 1),  we have proved that all sets of four ququad-ququad orthogonal maximally entangled states which are PPT-indistinguishable can be constructed in this way. Therefore,  the theorem has been proved.
\end{proof}

\centering {\bf (ii) Proof of Theorem 4 }
\begin{proof}
For case $t=3$, it suffices to proof that every 7 orthogonal lattice MESs in $\mathbb{C}^8\otimes\mathbb{C}^8$ are PPT-distinguishable. Similarly as what we have discussed in the proof of Theorem 3, we can list out $\sigma_n$'s for all columns in $N$ and also the $N'_{mn}$ for all rows corresponding to the variables in $\overrightarrow{r}_{\widehat{K}}^{(j)}(j \in \{1,2, \cdots, 7\})$ which equal to $0$ in $\overrightarrow{x_0}$ (in matrix form):
\begin{align*}
& \arraycolsep = 0.4em  \begin{array}{|c|c|c|c|c|}
\hline
    &  N_{\alpha_{\widehat{K}}}  &   N_{r_{K}}  & N_{\beta_{K}} & N_{\beta_{\widehat{K}}} \\
\hline
\sigma_n \text{'s} &  \overrightarrow{\mathbf{1}}^T  &  \overrightarrow{\mathbf{1}}^T  & \overrightarrow{\mathbf{0}}^T & -\overrightarrow{\mathbf{1}}^T  \\
\hline
\overrightarrow{r}_{\widehat{K}}^{(1)} & I & 0 & 0 & -I \\
\hline
\overrightarrow{r}_{\widehat{K}}^{(2)} & I & 0 & 0 & -I \\
\hline
\vdots & \vdots & \vdots & \vdots & \vdots \\
\hline
\overrightarrow{r}_{\widehat{K}}^{(7)} & I & 0 & 0 & -I \\
\hline
\end{array}\\
& \arraycolsep = 0.2em \begin{array}{|c|c|c|c|c|}
\hline
    &   N_{q^{(1)}} &  N_{q^{(2)}} & \cdots  &  N_{q^{(7)}} \\
\hline
\sigma_n \text{'s} &  P_{\overrightarrow{v}^{(1)}:}^{\otimes 3}  &  P_{\overrightarrow{v}^{(2)}:}^{\otimes 3} & \cdots  & P_{\overrightarrow{v}^{(7)}:}^{\otimes 3} \\
\hline
\overrightarrow{r}_{\widehat{K}}^{(1)} & -P_{\widehat{K}:}^{\otimes 3} & 0 & \cdots & 0 \\
\hline
\overrightarrow{r}_{\widehat{K}}^{(2)} & 0 & -P_{\widehat{K}:}^{\otimes 3} & \cdots &  0  \\
\hline
\vdots & \vdots & \vdots & \ddots & \vdots \\
\hline
\overrightarrow{r}_{\widehat{K}}^{(7)} & 0 & 0 & \cdots & -P_{\widehat{K}:}^{\otimes 3}   \\
\hline
\end{array}
\end{align*}
Also in this case, it is the structure of matrix $P^{\otimes 3}$ that affects, where
\begin{equation*}
P^{\otimes 3} = \frac12 \begin{bmatrix}
P^{\otimes 2} & P^{\otimes 2} & -P^{\otimes 2} & P^{\otimes 2} \\
P^{\otimes 2} & P^{\otimes 2} & P^{\otimes 2} & -P^{\otimes 2} \\
-P^{\otimes 2} & P^{\otimes 2} & P^{\otimes 2} & P^{\otimes 2} \\
P^{\otimes 2} & -P^{\otimes 2} & P^{\otimes 2} & P^{\otimes 2} \\
\end{bmatrix}.
\end{equation*}
The elements of $P^{\otimes 2}$ are either $\frac18$ or $-\frac18$, with 36 positive ones and 28 negatives at each row(or column). Moreover, for any two different rows,  $16$ places among those of the $36$ positive elements at one row take negative on the other row; Conversely,  $16$ places among those of the $28$ negative elements at one row take positive on the other row:
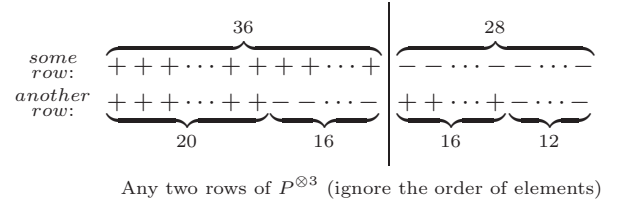
\begin{align*}
\begin{array}{cc@{}c|c@{}c}
                         &          \multicolumn{2}{c|}{_{36}}       &      \multicolumn{2}{c}{_{28}}       \\
\substack{some\\row:}    & \multicolumn{2}{c|}{\overbrace{+++\cdots++++\cdots+}} & \multicolumn{2}{c}{\overbrace{--\cdots--\cdots-}} \\
\substack{another\\row:} & \underbrace{+++\cdots++} & \underbrace{--\cdots-} & \underbrace{++\cdots+} & \underbrace{-\cdots-} \\
                         &    ^{20}    &      ^{16}          &      ^{16}          &    ^{12}     \\
\end{array} \\
\text{\scriptsize {Any two rows of $P^{\otimes 3}$ (ignore the order of elements)}}
\end{align*}
As what we have discussed in the proof of theorem 3, to add positive variables $x_n \in \overrightarrow{x_N}$ such that $\sum_{\{n|x_n>0\}} \sigma_n x_n < 0$ and $r_m^{(j)} =  \sum_{\{n|x_n>0\}} N'_{mn} x_n \geq 0 \ (j \in \{1, 2, \cdots, 7\}; m \notin K)$,  we need only to consider columns $N_{\alpha_{\widehat{K}}}$ and $N_q=[N_{q^{(1)}}, N_{q^{(2)}}, \cdots, N_{q^{(7)}}]$. Since $\sigma_n = 1$ for $n \in N_{\alpha_{\widehat{K}}}$,  we may choose also columns $n \in N_{q^{(j)}}(P_{\overrightarrow{v}^{(j)},  n}^{\otimes 3}<0)$ to satisfy $\sum_{\{n|x_n>0\}} \sigma_n x_n < 0$. Now suppose that we take $x_{n_{q^{(j)}}}>0$ for all those $n_{q^{(j)}}$'s such that $P_{\overrightarrow{v}^{(j)}, n_{q^{(j)}}}^{\otimes 3}<0 \ (j \in  \{1, 2, \cdots, 7\})$. Since each row of $P^{\otimes 3}$ has $28$ negative elements,  denoting these variables $x_{q_{1}^{(j)}}= \cdots =x_{q_{28}^{(j)}}=q^{(j)}>0 \ (j \in \{1, 2, \cdots, 7\})$,  we have
\begin{equation}
\sum_{m \notin K} x_{\alpha_{m}} < \frac18\sum_{j \in \{1, 2, \cdots, 7\}}28 \cdot q^{(j)}
\end{equation}
For each fixed $j \in \{1, 2, \cdots, 7\}$, on those $28$ columns,  $P^{\otimes 3}$ takes $(16+12-)$,  so we have
\begin{equation}
x_{\alpha_m}-\frac18\cdot 16q^{(j)}+\frac18\cdot 12q^{(j)} \geq 0 \ (m \notin K)
\end{equation}
and thus
\begin{equation}\begin{split}
\sum_{m \notin K} x_{\alpha_{m}} \geq (64-7) \cdot \frac18 \cdot & (16-12)q^{(j)} \\
 & (j \in \{1, 2, \cdots, 7\})
\end{split}\end{equation}
It's obvious that (29) contradicts (31). So we should take less $n_{q^{(j)}}$'s for those $P_{\overrightarrow{v}^{(j)}, n_{q^{(j)}}}^{\otimes 3}<0$ or take some $n_{q^{(j)}}$ such that $P_{\overrightarrow{v}^{(j)}, n_{q^{(j)}}}^{\otimes 3}>0$. As we what have already discussed in theorem 3,  the latter choice will be adverse, for that will decrease the right hand side of (29) while increase the right hand side of (31). Now suppose that $x_{q_{1}^{(j)}}= \cdots =x_{q_{27}^{(j)}}=q^{(j)}>0 \ (j \in \{1, 2, \cdots, 7\})$,  so
\begin{equation}
\sum_{m \notin K} x_{\alpha_{m}} < \frac18\sum_{j \in \{1, 2, \cdots, 7\}}27 \cdot q^{(j)}
\end{equation}
For any $27$ columns chosen like this for a particular $\overrightarrow{v}^{(j)}$,  36 rows of $P_{\widehat{\overrightarrow{v}^{(j)}}}^{\otimes 3}$ take $(15+12-)$,  27 rows of $P_{\widehat{\overrightarrow{v}^{(j)}}}^{\otimes 3}$ take $(16+11-)$,  because when we deselect $x_{q_{28}^{(j)}}$,  we decrease $1+$ from $(16+12-)$ for 36 rows and $1-$ from $(16+12-)$ for 27 rows in $P_{\widehat{\overrightarrow{v}^{(j)}}}^{\otimes 3}$. So for $m \notin K$ we have
\begin{equation}
 x_{\alpha_m} \geq \frac18\cdot 15q^{(j)}-\frac14\cdot 12q^{(j)}
\end{equation}
or
\begin{equation}
 x_{\alpha_m} \geq \frac18\cdot 16q^{(j)}-\frac14\cdot 11q^{(j)}
\end{equation}
As in the proof of theorem 3, columns corresponding to $x_{q_{28}^{(1)}}=x_{q_{28}^{(2)}}=\cdots=x_{q_{28}^{(7)}}=0$ in $P_{\overrightarrow{v}^{(1)}}^{\otimes 3}, P_{\overrightarrow{v}^{(2)}}^{\otimes 3}, \cdots, P_{\overrightarrow{v}^{(7)}}^{\otimes 3}$ respectively should be equal. In this case,
\begin{multline}
\sum_{m \notin K} x_{\alpha_{m}} \geq 36 \cdot \frac18 \cdot 3q^{(j)} + (28-7) \cdot \frac18 \cdot 5q^{(j)} \\
   (j \in \{1, 2, \cdots, 7)\}
\end{multline}
However (32) still contradicts (35). So we can further assume that $x_{q_{1}^{(j)}}= \cdots =x_{q_{26}^{(j)}}=q^{(j)}>0 \ (j \in \{1, 2, \cdots, 7\})$,  so
\begin{equation}
\sum_{m \notin K} x_{\alpha_{m}} < \frac18\sum_{j \in \{1, 2, \cdots, 7\}}26 \cdot q^{(j)}
\end{equation}
For the $26$ columns chosen like this for a particular $\overrightarrow{v}^{(j)}$,  20 rows of $P_{\widehat{\overrightarrow{v}^{(j)}}}^{\otimes 3}$ take $(14+12-)$,  32 rows of $P_{\widehat{\overrightarrow{v}^{(j)}}}^{\otimes 3}$ take $(15+11-)$ and 11 rows of $P_{\widehat{\overrightarrow{v}^{(j)}}}^{\otimes 3}$ take $16+10-$  because when we deselect $x_{q_{27}^{(j)}}$ and $x_{q_{28}^{(j)}}$,  we decrease $2+$ from $(16+12-)$ for 20 rows, $1+1-$ from $(16+12-)$ for 32 rows and $2-$ from $(16+12-)$ for 11 rows in $P_{\widehat{\overrightarrow{v}^{(j)}}}^{\otimes 3}$. So for $m \notin K$ we have
\begin{equation}
 x_{\alpha_m} \geq \frac18\cdot 14q^{(j)}-\frac18\cdot 12q^{(j)}
\end{equation}
or
\begin{equation}
 x_{\alpha_m} \geq \frac18\cdot 15q^{(j)}-\frac18\cdot 11q^{(j)}
\end{equation}
or
\begin{equation}
 x_{\alpha_m} \geq \frac18\cdot 16q^{(j)}-\frac18\cdot 10q^{(j)}
\end{equation}
To avoid the contrdiction like what between (29) and (31), we should hope more of $m \notin K$ to satisfies (37) or (38) rather than (39), therefore the columns corresponding to $x_{q_{28}^{(1)}}=x_{q_{28}^{(2)}}=\cdots=x_{q_{28}^{(7)}}=0$ in $P_{\overrightarrow{v}^{(1)}}^{\otimes 3}, P_{\overrightarrow{v}^{(2)}}^{\otimes 3}, \cdots, P_{\overrightarrow{v}^{(7)}}^{\otimes 3}$ respectively should be equal and so do the columns corrsponding to $x_{q_{27}^{(1)}}=x_{q_{27}^{(2)}}=\cdots=x_{q_{27}^{(7)}}=0$. In this case,
\begin{multline}
\sum_{m \notin K} x_{\alpha_{m}} \geq 20 \cdot \frac18 \cdot 2q^{(j)} + 32 \cdot \frac18 \cdot 4q^{(j)} \\
+ (12-7) \cdot \frac18 \cdot 6q^{(j)} \ (j \in \{1, 2, \cdots, 7\})
\end{multline}
Unfortunately, (36) and (40) are still contradicted. Henceforth, we may further deselect $x_{q_{26}^{(j)}}=x_{q_{27}^{(j)}}=x_{q_{28}^{(j)}}=0 \ (j \in \{1,2,\cdots,k\})$. However this is impossible as what we have discussed in Theorem 3, for any three columns of $P^{\otimes 3}$, only $4$ rows take ``$---$"(Note that if we are distinguishing $k=8$ states instead of 7 states, there won't exist contradiction between (36) and (40) or even between (32) and (35), which shows that some sets of 8 states are PPT-indistinguishable). In sum up, we have shown that for any 7 given rows $\{\overrightarrow{v}^{(1)}th, \overrightarrow{v}^{(2)}th, \cdots, \overrightarrow{v}^{(7)}th\}$, no variables $x_n > 0$ other than the original chosen $x_m$'s can be found such that equation (19) and (20) are satisfied, so all sets of 7 orthogonal lattice MESs are PPT-distinguishable.
\end{proof}

\centering {\bf (iii) Proof of Theorem 5 }
\begin{proof}
For case $t=4$, we can find variables $x_n's$ fulfilling (19) and (20) only when $k=16$ or $k=15$, as long as noticing the fact that for any four columns of $P^{\otimes 4}$, $24$ rows are like ``$----$", while for any five or more columns only less than $8$ rows are like ``$-----$". The other process is similar as those of Theorem 4 and the calculations are routine. Therefore any 14 orthogonal lattice MESs are PPT-distinguishable \cite{1}.
\end{proof}


\begin{thebibliography}{}
\bibitem{1}
The readers who are interested may investigate to find a general proof (that every $k<7d/8$ lattice MESs in $\mathbb{C}^{d}\otimes\mathbb{C}^{d}$ are PPT-distinghuishable, where $d=2^t$) suitable for all cases $t \geq 3$.
\end{thebibliography}
\end{document}